\documentclass[letterpaper]{article} % DO NOT CHANGE THIS
\usepackage{aaai23}  % DO NOT CHANGE THIS
\usepackage{times}  % DO NOT CHANGE THIS
\usepackage{helvet}  % DO NOT CHANGE THIS
\usepackage{courier}  % DO NOT CHANGE THIS
\usepackage[hyphens]{url}  % DO NOT CHANGE THIS
\usepackage{graphicx} % DO NOT CHANGE THIS
\usepackage{natbib}  % DO NOT CHANGE THIS AND DO NOT ADD ANY OPTIONS TO IT
\usepackage{caption} % DO NOT CHANGE THIS AND DO NOT ADD ANY OPTIONS TO IT
\title{Opinion Optimization  in Directed Social Networks}
\author {
    Haoxin Sun,
    Zhongzhi Zhang}%\thanks{Corresponding author. Zhongzhi Zhang is also with Shanghai Blockchain Engineering Research Center, as well as Research Institute of Intelligent Complex Systems, Fudan University, Shanghai 200433.

\affiliations {
     Shanghai Key Laboratory of Intelligent Information Processing, Fudan University, Shanghai 200433, China\\
     School of Computer Science, Fudan University, Shanghai 200433, China\\
    21210240097@m.fudan.edu.cn, zhangzz@fudan.edu.cn
}
\usepackage{amsmath}

\def\calG{\mathcal{G}}
\def\calF{\mathcal{F}}
\def\calR{\mathcal{R}}
\newcommand\DD{\boldsymbol{\mathit{D}}}
\newcommand\LL{\boldsymbol{\mathit{L}}}
\newcommand\II{\boldsymbol{\mathit{I}}}
\newcommand\sss{\boldsymbol{\mathit{s}}}
\renewcommand\AA{\boldsymbol{\mathit{A}}}
\newcommand\zz{\boldsymbol{\mathit{z}}}
\usepackage{bibentry}
\usepackage{amsmath}
\usepackage{amssymb}
\usepackage{bm}
\usepackage{nicefrac}
\usepackage{booktabs}
\usepackage{array}
\usepackage{multirow}
\usepackage{threeparttable}
\usepackage{makecell}
\usepackage[procnumbered,ruled,vlined,linesnumbered]{algorithm2e}
\usepackage{siunitx}
\usepackage{graphicx}
\usepackage{subfigure}
\usepackage{enumerate}

\newtheorem{problem}{Problem}
\newtheorem{theorem}{Theorem}[section]
\newtheorem{lemma}[theorem]{Lemma}

\newenvironment{fminipage}%
{\begin{Sbox}\begin{minipage}}%
		{\end{minipage}\end{Sbox}\fbox{\TheSbox}}

\DontPrintSemicolon
\SetKw{KwAnd}{and}
\SetFuncSty{textsc}
\SetKwInOut{Input}{Input\ \ \ \ }
\SetKwInOut{Output}{Output}

\usepackage{tabularx}
\begin{document}

\maketitle

\begin{abstract}
Shifting social opinions has far-reaching implications in various aspects, such as public health campaigns, product marketing, and political candidates. In this paper, we study a problem of opinion optimization based on the popular Friedkin-Johnsen (FJ) model  for opinion dynamics in an unweighted directed social network with $n$ nodes and $m$ edges. In the FJ model, the internal opinion of every node lies in the closed interval $[0, 1]$, with 0 and 1 being polar opposites of opinions about a certain issue. Concretely, we focus on the problem of selecting a small number of $ k\ll n $ nodes and changing their internal opinions to 0, in order to minimize the average opinion at equilibrium. We then design an algorithm that returns the optimal solution to the problem in $O(n^3)$ time. To speed up the computation, we further develop a fast algorithm by sampling  spanning forests, the time complexity of which is $ O(ln) $, with $l$ being the number of samplings. Finally, we execute extensive experiments on various real directed networks, which show that the effectiveness of our two algorithms is similar to each other, both of which  outperform several baseline strategies of node selection. Moreover, our fast algorithm is more efficient than the first one, which is scalable to massive graphs with more than twenty million nodes.
\end{abstract}

\section{Introduction}

As an important part of our lives, online social networks and social media have dramatically changed the way people propagate, exchange, and formulate opinions~\cite{Le20}. Increasing evidence indicates that in contrast to traditional communications and interaction, in the current digital age online communications and discussions have significantly influenced human activity in an unprecedented way, leading to universality, criticality and complexity of information propagation~\cite{NoCaFlMaRa22}. In order to understand mechanisms for opinion propagation and shaping, a variety of mathematical models for opinion dynamics have been established~\cite{JiMiFrBu15,PrTe17,DoZhKoDiLi18,AnYe19}. Among different models, the Friedkin-Johnsen (FJ) model~\cite{FrJo90} is a popular one, which has been applied to many aspects~\cite{BeWaVaHoShAl21,FrPrTePa16}. For example, the concatenated FJ model has been recently adapted to capture and reproduce the complex dynamics behind the Paris Agreement negotiation process, which   explains why consensus was achieved in these multilateral international negotiations~\cite{BeWaVaHoShAl21}.

A fundamental quantity for opinion dynamics is  the overall opinion or average opinion, which reflects the public opinions about  certain topics of interest. In the past years, the subject of modifying opinions in a graph has attracted considerable attention in the scientific community~\cite{GiTeTs13,AbKlPaTs18, ChLiSo19,XuHuWu20}, since it has important implications in diverse realistic situations, such as commercial marketing, political election, and public health campaigns. For example, previous work has formulated and studied the problem of  optimizing the overall or  average opinion for the FJ model in undirected graphs by changing a certain attribute of some chosen nodes, including internal opinion~\cite{XuHuWu20}, external opinion~\cite{GiTeTs13}, and susceptibility to persuasion~\cite{AbKlPaTs18, ChLiSo19}, and so on. Thus far, most existing studies about modifying opinions focused on undirected graphs. In this paper, we study the problem of minimizing or maximizing average opinion in directed graphs (digraphs), since they can better mimic realistic networks. Moreover, because previous algorithms for unweighted graphs do not carry over to digraphs, we will propose an efficient linear-time approximation algorithm to solve the problem.

%In the context of modifying opinions, the influence maximization problem aims to improve the overall favorable impression of a certain opinion item~\cite{GiTeTs13}. In recent years, the topic of efficiently modifying opinions in a social network has been an important direction of study in the field of opinion dynamics~\cite{GiTeTs13,AbKlPaTs18,XuHuWu20}, since the problem has far-reaching ramifications in a variety of situations, including commercial marketing, political campaigns, and public health efforts. Most of previous work focus on the operations on individualsin a social network so as to optimize different opinions, for example, to maximize or minimize the overall opinion of the network. However, solving these problems in large-scale graphs is a computational challenge, since it often involves matrix inversion and multiplication. Thus, in this paper, we will study the problems of minimizing the average of overall  steady-state opinions by changing the internal opinions, and give a fast linear time approximation algorithm to the problem.

We adopt the discrete-time FJ model in a social network modeled by a digraph $\calG=(V,E)$ with $ n $ nodes and $ m $ arcs. In the model, each node $i\in V$ is endowed with an internal/innate opinion $s_i$ in the interval $[0,1]$, where 0 and 1 are two polar opposing opinions regarding a certain topic. Moreover, each node $i\in V$ has an expressed opinion $z_i(t)$ at time $t$. During the opinion evolution process, the internal opinions of all nodes never change, while the expressed opinion $z_i(t+1)$ of any node $i$ at time $t+1$ evolves as a weighted average of $s_i$ and the expressed opinions of $i$’s neighbors at time $t$. For sufficiently large $t$, the expressed opinion $z_i(t)$ of every node $i$ converges to an equilibrium opinion $z_i$. We address the following optimization problem \textsc{OpinionMin} (or \textsc{OpinionMax}): Given a digraph $\calG=(V,E)$ and a positive integer $k\ll n$, how to choose $k$ nodes and change their internal opinions to $0$ (or 1), so that the average overall steady-state opinion is minimized  (or maximized).  

The main contributions of our work are as follows. We formalize the problem \textsc{OpinionMin} (or
\textsc{OpinionMax}) of optimizing the average equilibrium opinion by optimally selecting $k$ nodes and modifying their internal opinions to $0$ (or 1), and show that  both problems are equivalent to each other. We prove that the \textsc{OpinionMin} problem has an optimal solution and give an exact algorithm, which returns  the optimal solution in $ O(n^3)$ time. We then provide an interpretation for the average equilibrium opinion from the perspective of  spanning converging forests, based on which and Wilson's algorithm we propose a sampling based fast algorithm. The fast algorithm has an error guarantee for the main quantity concerned, and has a time complexity of  $O(ln)$, where $l$ is the number of samplings. Finally, we perform extensive experiments on various real networks, which shows that our fast algorithm is almost as effective as the exact one, both outperforming several natural baselines. Furthermore, compared with the exact algorithm, our fast algorithm is more efficient, and scales to massive graphs with more than twenty million nodes.

\section{Related Work}

In this section, we briefly review the existing work related to ours.

Establishing mathematical models is a key step for understanding opinion dynamics and various models have been developed in the past years~\cite{JiMiFrBu15,PrTe17,DoZhKoDiLi18,AnYe19}. Among existing models, the FJ model~\cite{FrJo90} is a classic one, which is a significant extension of the DeGroot model~\cite{De74}. Due to its theoretical and practical significance, the FJ model has received much interest since its development. A sufficient condition for stability of the FJ model was obtained in~\cite{RaFrTeIs15}, its average innate opinion was inferred in~\cite{DaGoPaSa13}, and the vector of its expressed opinions at equilibrium was derived in~\cite{DaGoPaSa13,BiKlOr15}. Moreover, some explanations of the FJ model were also provided~\cite{GhSr14,BiKlOr15}. Finally, in recent years many variants or extensions of the FJ model have been introduced and studied by incorporating different factors affecting opinion formation, such as peer pressure~\cite{SeGrSqRa19}, cooperation and competition~\cite{HeZhLi20,XuHuWu20}, and interactions among higher-order nearest neighbors~\cite{ZhXuZhCh20}.

In addition to the properties, interpretations and extensions of the FJ model itself, some social phenomena have been quantified based on the FJ model, such as disagreement~\cite{MuMuTs18}, conflict~\cite{ChLiDe18}, polarization~\cite{MaTeTs17,MuMuTs18},  and controversy~\cite{ChLiDe18}, and a randomized algorithm approximately computing polarization and disagreement was designed in~\cite{XuBaZh21}, which was later used in~\cite{TuNe22}. 
Also, many optimization problems for these quantities in the FJ model have been proposed and analyzed, including minimizing polarization~\cite{MuMuTs18,MaTeTs17}, disagreement~\cite{MuMuTs18},  and conflict~\cite{ChLiDe18,ZhZh22}, by different strategies such as modifying node's internal opinions~\cite{MaTeTs17}, allocating edge weights~\cite{MuMuTs18} and adding edges~\cite{ZhBaZh21}. In order to solve these problems, different  algorithms were designed by leveraging some mathematical tools, such as semidefinite programming~\cite{ChLiDe18} and  Laplacian solvers~\cite{ZhBaZh21}.

Apart from polarization, disagreement, and conflict, another important optimization objective for opinion dynamics is the overall opinion or average opinion at equilibrium. For example, based on the FJ model, maximizing or minimizing the overall opinion has been considered by using different node-based schemes, such as changing the node's internal opinions~\cite{XuHuWu20}, external opinions~\cite{GiTeTs13}, and susceptibility to persuasion~\cite{AbKlPaTs18, ChLiSo19}. On the other hand, for the DeGroot model of opinion dynamics in the presence of leaders, optimizing the overall opinion or average opinion was also heavily studied~\cite{VaFaFr14, YiCaPa21,ZhZh21}. An identical problem was also considered for a vote model~\cite{MeAsDaAmAn13}, the asymptotic mean opinion of which is similar to that in the extended DeGroot model~\cite{YiCaPa21}. The vast majority of previous studies concentrated on unweighted graphs, with the exception of a few works~\cite{AhDeHaMaYa15, YiCaPa21}, which addressed opinion optimization problems in digraphs and developed approximation algorithms with the time complexity of at least $O(n^{2.373})$. In comparison, our fast algorithm is more efficient since it has linear time complexity.

\section{Preliminary}
This section is devoted to a brief introduction to some useful notations and tools, in order to facilitate the description of  problem formulation and  algorithms.

\subsection{Directed Graph and Its Laplacian Matrix}

Let $\calG=(V,E)$ denote an unweighted simple directed graph (digraph) with $n=|V|$ nodes (vertices) and $m=|E|$ directed edges (arcs), where $V=\{v_1,v_2,\cdots,v_n\}$ is the set of nodes, and $E=\{(v_i, v_j)\in V \times V \}$ is the set of directed edges. The existence of arc $(v_i,v_j) \in E$ means that there is an arc pointing from node $v_i$ to node $v_j$. In what follows, $v_i$ and $i$ are used interchangeably to represent node $v_i$ if incurring no confusion. An isolated node is a node with no arcs pointing to or coming from it.  Let $N(i) $ denote the set of nodes that can be accessed by node $ i $. In other words, $N(i) =\{ j: (i,j)\in E\}$. A path $P$ from node $v_1 $ to $ v_k $ is an alternating sequence of nodes and arcs $v_1$,$(v_1,v_2)$,$v_2$,$\cdots$, $v_{j-1},(v_{j-1}$,$v_j)$, $v_j$ in which nodes are distinct and every arc $ (v_i,v_{i+1}) $ is from $ v_i $ to $ v_{i+1}$. A loop  is a path plus an arc from the ending node to the starting node. A digraph is (strongly) connected if for any pair nodes $v_x$ and $v_y$, there is a path from $v_x$ to $v_y$, and there is a path from $v_y$ to $v_x$ at the same time. A digraph is called weakly connected if it is connected when one replaces any directed edge $(i,j)$ with two directed edges $(i,j)$ and $(j,i)$ in opposite directions. A tree is a weakly connected graph with no loops. An isolated node is considered as a tree. A forest is a particular graph that is a disjoint union of trees.

The connections of digraph $\calG=(V,E)$ are encoded in its adjacency matrix $\AA=(a_{ij})_{n\times n}$, with the element $a_{ij}$ at row $i$ and column $j$ being $ 1 $  if $ (v_i,v_j) \in E $ and $a_{ij} = 0 $ otherwise. For a node $i$ in digraph $ \calG $, its  in-degree $d^+_i$ is defined as $d^+_i=\sum_{j=1}^n a_{ji}$, and its out-degree $d^-_i$ is defined as $d^-_i=\sum_{j=1}^n a_{ij}$. In the sequel, we use $d_i$ to represent the out-degree $d_i^-$. The diagonal out-degree matrix of digraph $\calG$ is defined as ${\DD} = {\rm diag}(d_1, d_2, \ldots, d_n)$, and the Laplacian matrix of digraph $\calG$ is defined to be ${\LL}={\DD}-{\AA}$. Let $\mathbf{1}$ and $\mathbf{0}$ be the two $n$-dimensional vectors with all entries being ones and zeros, respectively. Then, by definition, the sum of all entries in each row of $\LL$ is equal to $0$ obeying  $\LL\mathbf{1}=\mathbf{0}$. Let $\II$ be the $n$-dimensional identity matrix.

In a digraph $\calG$, if for any arc $(i,j)$, the arc $(j,i)$ exists, $\calG$ is reduced to an undirected graph. When $\calG $ is undirected, $a_{ij}= a_{ji}$ holds for an arbitrary pair of nodes $i$ and $j$, and thus $d^+_i= d^-_i$ holds for any  node $i\in V$. Moreover, in undirected graph $\calG $ both adjacency matrix $\AA$ and Laplacian matrix $\LL$ of  $\calG $ are symmetric, satisfying  $\LL\mathbf{1}=\mathbf{0}$.

\subsection{Friedkin-Johnsen Model on Digraphs}

The Friedkin-Johnsen (FJ) model~\cite{FrJo90} is a popular model  for opinion evolution and formation. For the FJ opinion model on a digraph $\calG=(V,E)$, each node/agent  $i\in V$ is associated with two opinions: one is the internal opinion $s_i$, the other is the expressed opinion $z_i(t)$ at time $t$. The internal opinion $s_i$ is in the closed interval $[0, 1]$, reflecting the intrinsic position of node $i $ on a certain topic, where 0 and 1 are polar opposites of opinions regarding the topic. A higher value of $s_i$ signifies that node $i$ is more favorable toward the topic, and vice versa. During the process of opinion evolution, the internal opinion $s_i$ remains constant, while the expressed opinion $z_i(t)$ evolves at time $ t+1$ as follows:
\begin{equation}\label{FJ}
z_i(t+1) = \frac{s_i +\sum_{j\in N(i)} a_{ij}z_j(t)}{1+\sum_{j\in N(i)} a_{ij}}.
\end{equation}
%In other words, at time $t+1$, the expressed opinion  $z_i(t+1)$ for node $i$ is influenced by both of its internal opinion $s_i$ and the expressed opinions of its neighbors at time $t$. %Note that we have made a common assumption in the literature that the weight of internal opinion of every node is unit.
 Let $ \sss = (s_1,s_2,\cdots,s_n)^\top$ denote the vector of internal opinions, and let $ \zz(t) = (z_1(t),z_2(t),\cdots,z_n(t))^\top $ denote the vector of expressed opinions at time $ t $. 
\begin{lemma}\cite{BiKlOr15}\label{le-z}
	As $ t $ approaches infinity, $ \zz(t) $ converges to an equilibrium vector  $ \zz = (z_1,z_2,\cdots,z_n)^\top$ satisfying $ \zz = (\II+\LL)^{-1}\sss $.
\end{lemma}

%The matrix $\left(\II+\LL\right)^{-1}  $ is called the fundamental
%matrix of opinion dynamics~\cite{GiTeTs13}. In the following, we denote the fundamental
%matrix by $ \mathbf{\Omega}=\left(\II+\LL\right)^{-1}=(\omega_{ij})_{n \times n}$.

%Lemma~\ref{le-z} shows that for each node $i$, its equilibrium expressed opinion $z_i$ is a linear combination of the internal opinions of all nodes. 
Let $\omega_{ij}$ be the element at the $i$-th row and the $j$-th column of matrix $\mathbf{\Omega} \triangleq \left(\II+\LL\right)^{-1}$, which is called the fundamental matrix of the FJ model for opinion dynamics~\cite{GiTeTs13}. The fundamental matrix has many good properties~\cite{ChSh97,ChSh98}. It is row stochastic, since $\sum_{j=1}^n \omega_{ij}=1$. Moreover, $ 0\leq\omega_{ji}< \omega_{ii}\leq 1$ for any pair of nodes $ i$ and $j $. The equality $ \omega_{ji} = 0$ holds if and only if $j\neq i$ and there is no path from node $ j $ to node $ i $; and $ \omega_{ii} = 1$ holds if and only if the out-degree $d_i$ of nodes $i$ is 0. Then, according to Lemma~\ref{le-z}, for every node $i \in V$, its expressed opinion $z_i$ is given by $z_i=\sum^n_{j=1}  \omega_{ij}s_j$, a convex combination of the internal opinions for all nodes. 
%\begin{equation}\label{zi}
%z_i=\sum^n_{j=1}  \omega_{ij}s_j.
%$\end{equation}

\section{Problem Formulation}

An important quantity for opinion dynamics is the overall expressed opinion or the average expressed opinion at equilibrium, the optimization problem for which on the FJ model has been addressed under different constraints~\cite{GiTeTs13,AhDeHaMaYa15,AbKlPaTs18,XuHuWu20,YiCaPa21}. In this section, we propose a problem of minimizing average expressed opinion for the FJ opinion dynamics model in a digraph, and design an exact algorithm optimally solving the problem.

\subsection{Average  Opinion and Structure Centrality}%Expressed

For the FJ model in digraph $\calG= (V,E)$, the overall expressed opinion is defined as the sum $z_{\rm sum}$ of expressed opinions $z_i$ of every node $i \in V$  at equilibrium. By Lemma~\ref{le-z}, $z_i=\sum^n_{j=1}  \omega_{ij}s_j$ and $ z_{\rm sum}=\sum_{i=1}^n z_i =\sum_{i=1}^n \sum_{j=1}^n\omega_{ji} s_i $. Given the vector for the equilibrium expressed opinions $\zz$, we use $g(\zz)$ to denote the average expressed opinion. By definition, 
\begin{equation}\label{g}
g(\zz) =\frac{1}{n} z_{\rm sum}=\frac{1}{n}\sum_{i=1}^n z_i =\sum_{i=1}^n\frac{\sum_{j=1}^n\omega_{ji}}{n}s_i\,.
\end{equation}
Since $g(\zz) = z_{\rm sum}/n$, related problems and algorithms for $g(\zz)$ and $z_{\rm sum}$ are equivalent to each other. In what follows, we focus on the quantity $g(\zz)$.

Equation~\eqref{g} tells us that the average expressed opinion $g(\zz)$ is determined by two aspects: the internal opinion of every node, as well as the network structure characterizing interactions between nodes encoded in matrix $\mathbf{\Omega}$, both of which constitute the social structure of opinion system for the FJ model. The former is an intrinsic property of each node, while the latter is a structure property of the network, both of which together determine the opinion dynamics system. Concretely, for the equilibrium expressed opinion $z_i=\sum_{j=1}^n \omega _{ij} s_j$ of node $i$, $\omega _{ij}$ indicates the convex combination coefficient or contribution of the internal opinion for node $j$.  And the average of the $j$-th column elements of $\mathbf{\Omega}$, denoted by $\rho_j \triangleq \frac{1}{n}\sum_{i=1}^n\omega_{ij}$, measures the contribution of the internal opinion of node $j$ to $g(\zz)$. We call $\rho_j$ as the structure centrality~\cite{Fr11} of node $j$ in opinion dynamics modelled by the FJ model, since it catches the long-run structure influence of node $j$ on the average expressed opinion. %By definition,
%\begin{equation}\label{rho}
%\rho_j = \frac{1}{n}\sum_{i=1}^n\omega_{ij}.
%\end{equation}
Note that matrix $\mathbf{\Omega}$ is row stochastic and $0\leq \omega_{ij} \leq 1$ for any pair of nodes $i$ and $j$, $ 0\leq \rho_j\leq 1 $ holds for every node $ j\in V$, and $\sum_{j=1}^n\rho_j=1$.

Using structure centrality, the average expressed opinion $g(\zz) $ is expressed as $g(\zz) = \sum_{i=1}^n\rho_{i}s_i$, which shows that the average expressed opinion $g(\zz) $ is a convex combination of the internal opinions of all nodes, with the weight for $s_i$ being the structure centrality $\rho_{i}$ of node $i$.

%In the past years, various metrics quantifying importance of nodes have been proposed, with popular ones including degree centrality~\cite{Ni74}, closeness centrality~\cite{Sa66},  betweenness centrality~\cite{Fr78},   eigenvector centrality~\cite{Bo72}, information centrality~\cite{StZe89}, PageRank~\cite{Gl15}, and forest distance closeness centrality~\cite{JiBaZh19}. These centrality measures have proven useful in a wide range of specific applications. However, since for a particular application scenario or dynamical process, an appropriate notion of node centrality measure depends on the network structure and the potential application objective, the aforementioned measures are not  appropriate to opinion dynamics. In the sequel, we will define a centrality measure, quantifying the influence of nodes in the formation of equilibrium expressed opinions.

%  We will show that the computation of our proposed structure centrality plays an important role in solving the problem and then give an optimal solution.

\subsection{Problem Statement}

As shown above, for a given  digraph $\calG= (V,E)$, its node centrality remains fixed. For the FJ model on $\calG= (V,E)$ with initial vector $ \sss = (s_1,s_2,\cdots,s_n)^\top$ of internal opinions, if we choose a set $T\subset V $ of $ k $ nodes and persuade them to change their internal opinions to 0, the average equilibrium opinion, denoted by $g_T(\zz)$, will decrease. It is clear that for $T=\emptyset$, $g_{\emptyset}(\zz)=g(\zz)$. Moreover, for two node sets $H$ and $T$, if $ T \subset H\subset V $, then $g_T(\zz)\geq g_H(\zz)$. Then the problem \textsc{OpinionMin} of opinion minimization arises naturally:  How to optimally select a set $T$ with a small number of $k$ nodes and change their internal opinions to 0, so that their influence on the overall equilibrium opinion is maximized. Mathematically, it is formally stated as follows.

\begin{problem}[OpinionMin]\label{Pr-IOMi}
	Given a digraph $ \calG = (V,E) $, a vector $ \sss $ of internal opinions, and an integer $ k\ll n $, we aim to find the set $ T \subseteq V $ with $ |T| = k $ nodes, and change the internal opinions of these chosen $ k $ nodes to $ 0 $, so that the average equilibrium opinions is minimized. That is,
	\begin{equation}\label{IOMi}
	T =  \arg \min_{U \subseteq V, |U|= k} g_U(\zz).
	\end{equation}
\end{problem}

Similarly, we can define the problem \textsc{OpinionMax} for maximizing the average equilibrium opinion by optimally selecting a set $T$ of $ k $ nodes and changing their internal opinions to 1. The goal of problem \textsc{OpinionMin} is to drive the average equilibrium opinion $g_T(\zz)$ towards the polar value 0, while the of goal of problem \textsc{OpinionMax} is to drive $g_T(\zz)$ towards polar value 1. Although the definitions and formulations of problems \textsc{OpinionMin} and \textsc{OpinionMax} are different, we can prove that they are equivalent to each other. %To see  the equivalence, we define a transformation $x \mapsto 1-x$ and apply it to the internal opinions and expressed opinions in every time step of opinion evolution described by~\eqref{FJ}. Since both the internal opinion and expressed opinion of every node are $[0, 1]$, it follows that~\eqref{FJ} still holds after the transformation. Hence, the \textsc{OpinionMax} problem in the original space of internal opinions is equivalent to the \textsc{OpinionMin} problem in the transformed space. 
In the sequel, we only consider the \textsc{OpinionMin} problem.

\subsection{Optimal Solution}

Although the \textsc{OpinionMin} problem is seemingly combinatorial, we next show that there is an exact algorithm optimally solving the problem in $ O(n^3) $ time.
\begin{theorem}\label{Th-IOMi}
	The optimal solution to the \textsc{OpinionMin} problem is the set $T$ of $k$ nodes with the largest product  of structure centrality  and internal opinion. That is, for any node $i \in T$ and any node $j \in V \setminus T$,  $ \rho_i s_i \geq  \rho_j s_j$. 
\end{theorem}
\begin{proof}
	Since the modifying of the internal opinions does not change the structure centrality $\rho_{i}$ of any node $i$, the optimal set $T$ of nodes for the \textsc{OpinionMin} problem satisfies
	\begin{equation*}
	\begin{aligned}
	T 
	= \arg \min_{U \subseteq V, |U|= k} \sum_{i \notin U}\rho_{i}s_i
	= \arg \max_{U \subseteq V, |U|= k} \sum_{i \in U}\rho_{i}s_i,
	\end{aligned}	
	\end{equation*}
	which finishes the proof.
\end{proof}

Theorem~\ref{Th-IOMi} shows that the key to solve \textsc{Problem}~\ref{Pr-IOMi} is to compute $\rho_i$ for every node $i$. In Algorithm~\ref{al-optimal}, we present an algorithm \textsc{Exact}, which computes $\rho_i$ exactly. The algorithm first computes the inverse $\mathbf{\Omega}$ of matrix $\II+\LL$, which takes $O(n^3)$ time. Based on the obtained $\mathbf{\Omega}=(\omega_{ij})_{n\times n}$, the algorithm then computes $\rho_is_i$ for each $i \in V $ in $ O(n^2) $ time, by using the relation $\rho_is_i =\frac{1}{n}\sum_{j=1}^n\omega_{ji}s_i $. Finally, Algorithm~\ref{al-optimal} constructs the set $T$ of $k$ nodes with the largest value of $\rho_is_i$, which takes $ O(kn) $ time. Therefore, the total time complexity of Algorithm~\ref{al-optimal} is $ O(n^3) $.  

Due to the high computation complexity, Algorithm~\ref{al-optimal} is computationally infeasible for large graphs. In the next section, we will give a fast algorithm for \textsc{Problem}~\ref{Pr-IOMi}, which is scalable to graphs with twenty million nodes.

%%%%%%%%%%%%%%%%%%%%%%%%%%%%%%%%%
%algorithm1
%%%%%%%%%%%%%%%%%%%%%%%%%%%%%%%%%
\begin{small}
\begin{algorithm}[tb]
	\caption{\textsc{Exact}$(\calG,\sss, k)$}
	\label{al-optimal}
	\Input{
		A digraph $\calG=(V,E)$; an internal opinion vector $ \sss $; an integer $k$ obeying relation $1 \leq k \leq |V|$\\
	}
	\Output{
		$T$: A subset of $V$ with $|T|=k$
	}
	Initialize solution $T= \emptyset$ \;
	Compute $\mathbf{\Omega} = (\II+\LL)^{-1}$ \;
	Compute $\rho_is_i =\frac{1}{n}\sum_{j=1}^n\omega_{ji}s_i$ for each $i \in V $\;
	\For{$ t=1 $ to $ k $}{
		Select $i$ s. t.  $i \gets \mathrm{arg\, max}_{i \in V \setminus T} \rho_is_i$ \;
		Update solution $T \gets T\cup \{ i \}$ \;	
	}	
	\Return $T$.
\end{algorithm}
\end{small}
%%%%%%%%%%%%%%%%%%%%%%%%%%%%%%%%%

% Using Lemma \ref{le-omega}, we can get $ \frac{1}{n^2} \leq \rho_{i}\leq \frac{n+1}{2n} $. $ \rho_i = \frac{1}{n^2} $ if and only if $ (j,i) \notin E $ and $ (i,j) \in E $, $ \forall j\in V $. $ \rho_i = \frac{n+1}{2n} $ if and only if $ (j,i) \in E $ and $ (i,j) \notin E $, $ \forall j\in V $.  

\section{Fast Sampling Algorithm}

In this section, we develop a linear time algorithm to approximately evaluate the structure centrality of every node and solve the \textsc{OpinionMin} problem by using the  connection of the fundamental matrix $\mathbf{\Omega}$ and the  spanning converging forest.  Our fast algorithm is based on the sampling of  spanning converging forests, the ingredient of which is an extention of  Wilson’s algorithm~\cite{Wi96,WiPr96}.

\subsection{Interpretation of Structure Centrality } %in Terms of Spanning Converging Forests

%In this subsection, we provide an explanation for the structure centrality from the perspective of spanning converging forests.

For a digraph $\calG=(V,E)$, a spanning subgraph of $\calG$  is a subgraph of  $\calG$ with node set being $V$ and edge set being a subset of $E$. A converging tree is a weakly connected digraph, where one node, called the root node, has out-degree 0 and all other nodes have out-degree 1. An isolated node is considered as a converging tree with the root being itself. A spanning converging forest of digraph $\calG$ is a spanning subdigraph of $\calG$, where all weakly connected components are converging trees. A spanning converging forest is in fact an in-forest in~\cite{AgCh01,ChAg02}. 

Let $\calF $ be the set of all spanning converging forests of digraph $ \calG $. For a spanning converging forest $\phi \in \mathcal{F}$, let $ V_{\phi} $ and $E_{\phi} $ denote its node set and arc set, respectively. By definition, for each node $ i \in V_{\phi} $, there is at most one node $ j \in V_{\phi} $ obeying $ (i,j) \in E_{\phi}$. For a spanning converging forest $\phi$, define $\mathcal{R}(\phi ) = \{i:(i,j) \notin \phi, \forall j \in V_{\phi} \}$, which is actually the set of roots of all converging trees that constitute $ \phi$. Since each node $i$ in $\phi$ belongs to a certain converging tree, we define function $r_{\phi}(i): V \rightarrow \calR(\phi) $ to map node $i$ to the root of the converging tree including $i$. Thus, if $ r_{\phi}(i) = j $ we conclude that $ j\in \calR(\phi) $, and nodes $ i$ and $j $ belong to the same converging tree in $ \phi $. Define $ \calF_{ij} $ to be the set of those spanning converging forests, where for each spanning converging forest nodes $ i$ and $j $ are in the same converging tree rooted at node $ j $. That is, $\calF_{ij} = \{\phi: r_{\phi}(i) = j, \phi \in \calF\}$. Then, we have $\calF_{ii} = \{\phi: i\in \calR(\phi), \phi \in \calF\}$.

% For a better understanding of  the related notions about spanning converging forests, we present an example in Figure~\ref{fg1}. For the toy digraph in Figure~\ref{fg1}, there are 27 spanning converging forests, among which the 9 spanning converging forests with cyan background belong to $\calF_{24}$.

% %%%%%%%%%%%%%%%%%%%%%%%%%	
% \begin{figure}[tb]
% 	\centering
% 	\includegraphics[width=.9\linewidth]{fg1.eps}
% 	\caption{A digraph and its all spanning converging forests. The nine  spanning converging  forests   in  set $ \calF_{24} $ are marked in cyan. \label{fg1}}
% \end{figure}
%%%%%%%%%%%%%% %%%%%%%%%%%%%

Spanning converging forests have a close connection with the fundamental matrix of the FJ model, which is in fact the in-forest matrix of a digraph $\calG$ introduced~\cite{ChSh97,ChSh98}. Using the approach in~\cite{Ch82}, it is easy to derive that the entry $\omega_{ij}$ of the fundamental matrix $\mathbf{\Omega}$ can be written as $\omega_{ij}= |\calF_{ij}|/|\calF|$.
%\begin{equation}\label{omegaij}
%\omega_{ij}= \frac{|\calF_{ij}|}{|\calF|}.
%\end{equation}
% Then, according to Figure~\ref{fg1}, the fundamental matrix of the digraph in Figure~\ref{fg1} is
% \begin{equation}\label{Omega}
% 	\mathbf{\Omega}=\left(\II+\LL\right)^{-1}=	\begin{pmatrix}
% 		 0.444 &0.148 & 0.074 & 0.333\\
% 		0.111 & 0.370   & 0.185   &0.333\\
% 		0.111 &0.037  & 0.519 &  0.333\\
% 		0.222& 0.074 & 0.037 &  0.667
% 	\end{pmatrix}.
% \end{equation}
% Since there are nine spanning  converging forests in $\calF_{24} $,  we have $\omega_{24} =|\calF_{24}|/|\calF| = 9/27 = 0.333$. Using~\eqref{Omega}, for the four nodes of the digraph in Figure~\ref{fg1}, the values of structure centrality are $\rho_1=0.222$, $\rho_2=0.157$, $\rho_3=0.204$, and $\rho_4=0.417$, respectively. 

%In figure \ref{fg1}, we have $ \rho_1 = 0.889 $, $ \rho_2 = 0.630 $, $ \rho_3 = 0.819 $, $ \rho_4 = 1.667 $. $ \rho_4 $ is maximum, which is in agreement with our intuition since

%The fundamental matrix $\mathbf{\Omega}$ has many good properties~\cite{ChSh97,ChSh98}. It is row stochastic, since $\sum_{j = 1}^n \omega_{ij}=\sum_{j = 1}^n |\calF_{ij}|/|\calF|=1$. Moreover, $ 0\leq\omega_{ji}< \omega_{ii}\leq 1$ for any pair of nodes $ i$ and $j $. The equality $ \omega_{ji} = 0$ holds if and only if $ j\neq i $ and there is no path from node $ j $ to node $ i $; and $ \omega_{ii} = 1$ holds if and only if the out degree $d_i$ of nodes $i$ is 0.

With the notions mentioned above, we now provide an interpretation and another expression of structure centrality $\rho_i$ for any node $i$. For the convenience of description, we introduce some more notations. For a node $i \in V$ and a  spanning converging forest $\phi \in \calF $ of digraph $\calG= (V,E)$, let $M(\phi,i) $ be a set defined by $M(\phi,i) = \{ j: r_\phi(j) = i \}$. By definition, for any  $\phi \in \calF $, if $ i \notin \calR(\phi) $, $ M(\phi,i) = \emptyset $; if $ i \in \calR(\phi) $, $ |M(\phi,i)| $ is equal to the number of nodes in the converging tree in $\phi$, whose root is node $i$. For two nodes $i$ and $j$ and a  spanning converging forest $\phi$, define $s(\phi,j,i)$ as a function taking two values, 0 or 1:
\begin{equation}
s(\phi,j,i) = \begin{cases}
1 & \text{ if } r_{\phi}(j)= i,\\
0 & \text{ if } r_{\phi}(j)\neq i.
\end{cases}
\end{equation}
Then, the structure centrality $\rho_i $ of node $i$  is recast as
% \begin{equation}
\begin{align}\label{rhonew}
\rho_i &= \frac{1}{n}\sum_{j=1}^n\omega_{ji}  = \frac{1}{n|\calF|}\sum_{j=1}^n  |\calF_{ji}|
= \frac{1}{n|\calF|}\sum\limits_{j=1}^{n}\sum\limits_{\phi\in \calF}s(\phi,j,i)	\notag\\ 
& = \frac{1}{n|\calF|}\sum\limits_{\phi\in \calF}\sum\limits_{j=1}^{n}s(\phi,j,i)
= \frac{1}{n|\calF|}\sum\limits_{\phi\in \calF}|M(\phi,i)|,
\end{align}
% \end{equation}
which indicates that $\rho_i$ is the average number of nodes  in the converging trees rooted at node $i$ in all $\phi\in \calF$, divided by $n$.

\subsection{An Expansion of Wilson's Algorithm}
%Our work is related to a  particular distribution on spanning
%trees called random spanning trees. 
%Let $ w(i,j) $ be the weight of edge $ (i,j) $. In unweight graph, $ w(i,j) = 1  $ if $ (i,j) \in E $, $ w(i,j) =0  $ otherwise. For $ \tau \in \calT $, let $ w(\tau) = \prod_{(i,j)\in \tau} w(i,j) $ be the weight of $ \tau $. 
%A random spanning tree $ \tau_0 $ is a randomly generated spanning tree from the following distribution:
%\begin{equation}\label{RST}
%	\mathbb{P}(\tau = \tau_0)  = \frac{1}{|\calT|}.
%\end{equation}
%
%In unweighted graph, the distribution is uniform over all possible spanning trees. 
%
%Just like random spanning trees, in unweighted digraph,
% let $ w(\phi) =  \prod_{(i,j)}\in \phi$ be the weight of a rooted spanning forest $ \phi $. 
%a random spanning forest $ \phi_0$ is a randomly generated spanning forest from the following distribution:
%\begin{equation}\label{RSF}
%	\mathbb{P}(\phi = \phi_0)   = \frac{1}{|\calF|}.
%\end{equation}

%Here we introduce an extension of Wilson's algorithm, which is a basis of our proposed approximation algorithm for the \textsc{OpinionMin} problem. 
We first give a brief introduction to the loop-erasure operation to a random walk~\cite{La80}, which is a process obtained from the random walk by performing an erasure operation on its loops in chronological order. Concretely, for a random walk $P=v_1,(v_1,v_2),v_2,\ldots,v_{k-1},(v_{j-1},v_j),v_j$, %we use $V_P= (v_1,v_2,\cdots,v_j)$ to denote the nodes in path $P$ in order, that is, the first $j+1$ nodes visited the walk. Then, 
the loop-erasure $P_{\rm LE}$ to $P$ is an alternating
sequence  $\widetilde{v}_1,(\widetilde{v}_1, \widetilde{v}_2), \widetilde{v}_2\ldots, \widetilde{v}_{q-1},(\widetilde{v}_{q-1}, \widetilde{v}_q),\widetilde{v}_q$ of nodes and arcs  obtained inductively in the following way. First set $\widetilde{v}_1= v_1$ and append $ \widetilde{v}_1$ to $P_{\rm LE}$. Suppose that sequence $\widetilde{v}_1$, $(\widetilde{v}_1, \widetilde{v}_2)$, $\widetilde{v}_2$, $\ldots$, $\widetilde{v}_{h-1}$, $(\widetilde{v}_{h-1},\widetilde{v}_h)$, $\widetilde{v}_h$ has been added to $P_{\rm LE}$ for some $h\geq 1$. If $\widetilde{v}_h=v_j$, then $q= h$ and $\widetilde{v}_h$ is the last node in $P_{\rm LE}$. Otherwise, define $ \widetilde{v}_{h+1}= v_{r+1}$, where $ r = \max\{i:v_i = \widetilde{v }_h \}$.

%Wilson proposed a algorithm  based on loop-erased random walk to get a  spanning tree rooted at a given node~\cite{Wi96}. For a path $ P $, its loop erasure~\cite{LaFr79} is a simple path created by removing all cycles of $ P  $ in chronological order. More precisely, given a path $ P  =  v_1,(v_1,v_2),v_2,\cdots, v_{k-1},(v_{j-1},v_j), v_j $. We use $ V_P = (v_1,v_2,\cdots,v_j) $ to denote the nodes in path $ P $ in order, which we can think of as the first $ j + 1  $ nodes visited by some randomwalk, we define the loop-erasure $ LE(P) $ to be the sequence $ (\widetilde{v}_1,\cdots,\widetilde{v}_q) $ obtained inductively.First set $ \widetilde{v}_1 = v_1 $ and append $ \widetilde{v}_1 $ to $ LE(P) $. Suppose now $ \widetilde{v}_1,\cdots,\widetilde{v}_h $ have been added to  $ LE(P) $ for some $ h\geq 1 $. If $ \widetilde{v}_h  = v_j $, then $ q = h $ and $ \widetilde{v}_h $ is the last node in the sequence $ LE(P) $.  Otherwise, define $ \widetilde{v}_{h+1}  = v_{r+1} $, where $ r = \max\{i:v_i = \widetilde{v_h}  \} $.	

%Wilson proposed a algorithm  based on loop-erased random walk to get a  spanning tree rooted at a given node~\cite{Wi96}. Following the  steps below,  we will give a brief introduction of Wilson's algorithm~\cite{Wi96} to get a spanning tree $\tau $ rooted at node $ u $. 

Based on the loop-erasure operation on a random walk, Wilson proposed an algorithm to generate a uniform spanning tree rooted at a given node~\cite{Wi96,WiPr96}. Following the three steps below, we introduce  Wilson's algorithm to get a spanning tree $\tau=( V_{\tau}, E_{\tau})$ of a connected digraph $\calG= (V,E)$, which is rooted at node $u$. (i) Set $ \tau  =(\{u\},\emptyset) $ with $ V_{\tau}  = \{u\} $. Choose $ i \in V \setminus V_{\tau} $. Then create an unbiased random walk starting at node $i$. At each time step, the walk jumps to a neighbor of current position with identical probability. The walk stops, when the whole walk $P$ reaches some node in $\tau $. (ii) Perform loop-erasure operation on the random walk $P$ to get $P_{\rm LE}=\widetilde{v}_1$, $(\widetilde{v}_1,\widetilde{v}_2)$, $\widetilde{v}_2$, $\ldots$, $(\widetilde{v}_{q-1},\widetilde{v}_q)$, $\widetilde{v}_q$, and add the nodes and arcs in $P_{\rm LE}$ to $\tau$. Then update $V_{\tau} $ with the nodes in $ \tau $. (iii) If $V_{\tau}  \neq V $, repeat step (ii), otherwise end circulation and return $\tau$.

For a digraph $\calG= (V,E)$, connected or disconnected, we can also apply Wilson's Algorithm to get a spanning converging forest $\phi_0 \in \calF$, by using the method similar to that in~\cite{AvLuGaAl18,PiAmBaTr21}, which includes the following three steps. (i) We construct an augmented digraph $\calG'=(V',E') $ of $\calG= (V,E)$, obtained from $\calG= (V,E)$ by adding a new node $\Delta$ and some new edges. Concretely, in $\calG'=(V',E')$, $V' =V \cup \{\Delta\}$ and $E'= E\cup \{(i,\Delta)\} \cup \{(\Delta,i)\}$ for all $i\in V$. (ii) Using Wilson's algorithm to generate a uniform spanning tree $\tau$ for the augmented graph $\calG'$, whose root node is $\Delta$. (iii) Deleting all the edges $(i,\Delta) \in \tau$ we get a spanning forest $\phi_0 \in \calF$ of $\calG$. Assigning $\calR(\phi_0) = \{ i:(i,\Delta)\in \tau \}$ as the set of roots for trees $ \phi_0$ makes $\phi_0$ become a converging   spanning forest of $\calG$.

The spanning converging forest $\phi_0$ obtained using the above steps is uniformly selected from $\calF$~\cite{AvLuGaAl18}. In other words, for any spanning converging forest $\phi $ in $\calF$, we have $\mathbb{P}(\phi_0 =\phi) = 1/|\calF|$. Following the three steps above for generating a uniform spanning converging forest of digraph $\calG$, in Algorithm 2 we present an algorithm to generate  a uniform spanning converging forest $\phi$ of digraph $\calG$, which returns a list RootIndex with the $i$-th element RootIndex[$i$]  recording the root of the tree in $\phi$ node $ i $ belongs to.  That is, RootIndex[$i$]=$r_{\phi}(i)$. %Note that Algorithm~\ref{alg-rf} is an extension of that in~\cite{PiAmBaTr21}, by  recording and returning the root index of every node.

%%%%%%%%%%%%%%%%%%%%%%%%%%%%%%%%%%%%%%%%%%%%%%%	
% Algorithm 2
%%%%%%%%%%%%%%%%%%%%%%%%%%%%%%%%%%%%%%%%%%%%%%%	
\begin{small}
\begin{algorithm}
	\caption{$\textsc{RandomForest}(\calG)$}
	\label{alg-rf}
	\Input{ $\calG$ : a digraph  \\
	}
	\Output{RootIndex : a vector recording the root index of every node }
	InForest[$ i $] $ \leftarrow $ false , $i= 1,2,\ldots,n $\;
	Next[$ i $] $ \leftarrow -1$ , $i= 1,2,\ldots,n $\;
	RootIndex[$i$] $ \leftarrow 0 $, $i= 1,2,\ldots,n $\;
	\For{$ i = 1 $ to $ n $}
	{$ u \leftarrow i $\; 
		\While{not InForest[$ u $]}{
			seed $ \leftarrow $ \textsc{Rand}()  \; 
			\If{seed  $\leq  \frac{1}{1+d_u} $}{
				InForest[$ u $] $ \leftarrow $ true\;
				Next[$ u $]$ \leftarrow -1 $\;
				RootIndex[$ u $] $ \leftarrow u $\;
			}
			\Else{
				Next[$ u $] $ \leftarrow $ \textsc{RandomSuccessor}($ u,\calG $)\; 
				u $ \leftarrow $ Next[$ u $]\;
			}
		}
		RootNow $ \leftarrow $ 	RootIndex[$ u $] \;
		$ u\leftarrow i $\;
		\While{not InForest[$ u $]}{
			InForest[$ u $] $ \leftarrow $ true\;
			RootIndex[$ u $] $ \leftarrow $ RootNow\;
			u $ \leftarrow $ Next[$ u $]\;
		}
	}
	
	\textbf{return} RootIndex \;
\end{algorithm}
\end{small}
%%%%%%%%%%%%%%%%%%%%%%%%%%%%%%%%%%%%%%%%%%%%%%%	

 Below we give a detailed description for Algorithm~\ref{alg-rf}. InForest is a list recording whether a node is in the forest or not in the random walk process. In line 1, we initialize InForest$ [i] $ to false, for all $ i\in V$. If node $i$ is not a root of any tree in the forest $\phi$, Next[$i$] is the node $j$ satisfying $(i,j) \in \phi$; if node $i$ belongs to the root set $\calR(\phi) $, Next[$ i $] $ = -1$. We initialize Next[$ i $] $= -1$ in line 2. We start a random walk at node $u$ in the extended graph $\calG'$ in line 5 to create a forest branch of $\calG$. The probability of visiting node $\Delta$ starting from $u$ is $1/(1+d_u)$. In line 7, we generate a random real number in $(0,1)$ using function \textsc{Rand}(). If the random number satisfies the inequality in line 8, the walk jumps to node $\Delta$ at this step. According to the previous analysis, in extended graph $ \calG'$, those nodes that point directly to $  \Delta $  belong to the root set $ \calR(\phi) $. 
 In lines 9-11, we set the node $u$ as a root node and update InForest[$u$], Next$[u]$, and RootIndex$[u] $. If the inequality in line 8 does not hold, we use function \textsc{RandomSuccessor}($u,\calG$) to return a node randomly selected from the neighbors of $ u $ in $\calG $ in line 13. Then we update $ u $ to Next[$u $] and go back to line 6. The $\mathbf {for}$ loop stops when the random walk goes to a node already existing in the forest. When the loop stops, we get a newly created branch. In lines 15-20, we add the loop-erasure of the branch to the forest and then update RootIndex.

We now analyze the time complexity of Algorithm~\ref{alg-rf}. Before doing this, we present some properties of the diagonal element $\omega_{ii}$ of matrix $\mathbf{\Omega}$ for all nodes $i \in V$.
\begin{lemma}\label{le-omega}
	For any $i=1,2,\ldots, n$, the diagonal element $\omega_{ii}$ of matrix $\mathbf{\Omega}$ sastisfies $\frac{1}{1+d_i}\leq \omega_{ii} \leq \frac{2}{2+d_i}$.
\end{lemma}	
% \begin{proof}
% 	Let $l_{ij} $ be the entry at row $ i $ and column $ j $ of the Laplacian matrix $\LL $. Then $l_{ii} = d_i $, and $ l_{ij} = 0 $ or $ -1 $ for $i \neq j$. Using $\mathbf{\Omega}\left(\II+\LL\right)  = \II$,  one obtains that for any $i=1,2,\ldots, n$,
% 	\begin{equation*}\label{key}
% 	1 = (1+d_i)\omega_{ii}+\sum_{u \neq i} \omega_{iu}l_{ui}\leq (1+d_i)\omega_{ii},
% 	\end{equation*}
% 	which implies that $ \omega_{ii}\geq \frac{1}{1+d_i} $ with  the equality holding  under certain conditions, for example,    $(u,i) \notin E $ for arbitrary $ u \in V $.
% 	On the other hand, considering the fact that $ \sum_{j = 1}^n \omega_{ij} =1 $, one has 
% 	\begin{equation*}
% 	1 =(1+d_i)\omega_{ii}+\sum_{u \neq i} \omega_{iu}l_{ui}\geq (1+d_i)w_{ii}-(1-w_{ii}),
% 	\end{equation*}
% 	which means  $\omega_{ii} \leq \frac{2}{2+d_i}$. The equality holds  under some conditions, for example,  $ (u,i) \in E $ for any $u \in V$.
% \end{proof}	

\begin{lemma}\label{le-rf}
	For any unweighted digraph $ \calG = (V,E) $, the expected time complexity of Algorithm~\ref{alg-rf} is $ O(n) $.
\end{lemma}

\begin{proof}
	Wilson showed that the expected running time of generating a uniform spanning tree of a connected digraph $\calG$ rooted at node $u$ is equal to a weighted average of commute time between the root and the other nodes~\cite{Wi96}. Marchal rewrote this average of commute time in terms of graph matrices in Proposition 1 in~\cite{Ma00}. According to Marchal's result, the expected running time of Algorithm 2 is equal to the trace $\sum_{i=1}^n \omega_{ii}(1+d_i)$ of matrix $\mathbf{\Omega}(\II+\DD)$. Using Lemma~\ref{le-rf}, we have
	$\sum_{i=1}^n \omega_{ii}(1+d_i) \leq \sum_{i=1}^n \frac{2+2d_i}{2+d_i} \leq 2n\left(1-\frac{1}{n+1}\right).$ Thus, the expected time complexity of Algorithm 2 is $O(n)$.
\end{proof}

\subsection{Fast Approximation Algorithm}

%Theorem~\ref{th-Fast} shows the existence of an optimal solution to the problem \textsc{OpinionMin}, which can be obtained by Algorithm~\ref{al-optimal} in $O(n^3)$ time. This cube complexity makes Algorithm~\ref{al-optimal} difficult to implement in larger graphs. In Algorithm~\ref{al-optimal}, the most time-consuming step is to compute the inverse $\mathbf{\Omega}$ of matrix $\II+\LL$, in order to obtain $\rho_i$ for every node $i \in V$. Since $\mathbf{\Omega}$ is a dense matrix, even if we  know all the entries of matrix $\mathbf{\Omega}$ in advance, computing $\rho_i$ through $\rho_i=\sum_{j=1}^n\omega_{ji}/n$ still needs $O(n^2) $ time for all nodes $i \in V$. 
Here by using~\eqref{rhonew}, we present an efficient sampling-based algorithm \textsc{Fast} to estimate $\rho_i$ for all $i \in V$ and approximately solve the problem \textsc{OpinionMin} in linear time. 

The ingredient of the approximation algorithm \textsc{Fast} is the variation of Wilson's algorithm introduced in the preceding subsection. The details of algorithm \textsc{Fast} are described in Algorithm~\ref{al-Fast}. First, by applying   Algorithm~\ref{alg-rf} we generate $l$ random spanning converging forests $\phi_1,\phi_2,\ldots,\phi_l$. Then, we compute $\widehat{\rho}_i = \frac{1}{nl}\sum_{j=1}^l|M(\phi_j,i)|$ for all $i \in V$. Note that each of these $l$ spanning converging forests has the same probability of being created from all spanning converging forests in $\calF$~\cite{AvLuGaAl18}. Thus, we have $\mathbb{E} \left(\frac{1}{nl}\sum_{j=1}^l\left| M(\phi_j,i)\right| \right)= \rho_i$, which implies that $\widehat{\rho}_i$ in Algorithm~\ref{al-Fast} is an unbiased estimation of $\rho_i$. Then, 
$\widehat{\rho}_i s_i$ is an unbiased estimation of $\rho_is_i$. Finally, we choose $k$ nodes from $V$ with the top-$k$  values of $\widehat{\rho}_is_i$.

%%%%%%%%%%%%%%%%%%%%%%%%%%%%%%%%%%%%%%%%%%%%%%%	
% Algorithm3
%%%%%%%%%%%%%%%%%%%%%%%%%%%%%%%%%%%%%%%%%%%%%%%	
\begin{small}
\begin{algorithm}[t]	\caption{$\textsc{Fast}\left(\calG,k,l\right)$}
	\label{al-Fast}
	\Input{$\calG$ : a digraph 	\\$ k $ : size of the target set \\	$l $ :  number of generated spanning forests \\
	}
	\Output{$ \widehat{T} $ : the target set }
	\textbf{Initialize} : $ \widehat{T} \leftarrow \emptyset $, $ \widehat{\rho}_i \leftarrow 0,\ i=1,2,\ldots,n $\;
	\For{$t = 1$ to $l$}{
		RootIndex $ \leftarrow $ \textsc{RandomForest}($ \calG $)\\
		\For{$ i=1 $ to $ n $}{
			$ u \leftarrow $ RootIndex[$i$]\\
			$\widehat{\rho}_u \leftarrow \widehat{\rho}_u +1$}
	}
	$\widehat{\rho} \leftarrow \widehat{\rho}/nl$ \qquad\qquad\% $\widehat{\rho}=(\widehat{\rho}_1,\widehat{\rho}_2,\ldots,\widehat{\rho}_n)^\top$\\
	\For{$i = 1$ to $k$}{
		$ u \leftarrow \arg \max\limits_{q\in V\setminus \widehat{T}} \widehat{\rho}_q s_q $\\
		$ \widehat{T} \leftarrow \widehat{T} \bigcup \{u\} $\\
	}
	\textbf{return} $ \widehat{T} $ \;
\end{algorithm}
\end{small}
%%%%%%%%%%%%%%%%%%%%%%%%%%%%%%%%%%%%%%%%%%%%%%%

\begin{theorem}\label{th-time}	
	The time complexity of Algorithm 3 is $O(ln)$.
\end{theorem}

%\begin{proof}
%	First, according to Lemma~\ref{le-rf}, generating $l$ random spanning converging forests takes $O(ln) $ time. Second, calculating $\widehat{\rho}_i$ also needs $ O(ln) $ time. Finally, choosing $k$ nodes with the largest value of $ \widehat{\rho}_is_i$ takes $ O(kn) $ time. Thus, the total time complexity of Algorithm 3 is $ O(ln)$, which finishes the proof.
%\end{proof}

Running Algorithm 2 requires determining the number of sampling $ l $, which determines the accuracy of $\widehat{\rho}_is_i$ as an approximation of $\rho_is_i $. In general, the larger the value of $l$, the more accurate the estimation of $\widehat{\rho}_is_i$ to $ \rho_is_i$. Next, we bound the number $l$ of required samplings of spanning converging forests to guarantee a desired estimation precision of $\widehat{\rho}_is_i$ by applying the Hoeffding's inequality~\cite{Ho63}.

%\begin{lemma}\label{le-Ho} (Hoeffding's Inequality~\cite{Ho63}). Let $ X_1, X_2,\cdots,X_l $ be  independent random variables belonging to intertal $ [a,b ]$ and satisfying $\mathbb{E}(X_j)= \mu$, for $j = 1,2,\ldots, l$. Then, for any $ \epsilon >0 $,
%	\begin{equation}
%		\mathbb{P}\left\{\left| \frac{1}{l}\sum_{j = 1}^l X_j -\mu \right|>\epsilon \right\} \leq 2e^{-2l\epsilon^2/(b-a)^2}.
%	\end{equation}  
%\end{lemma}

We now demonstrate that with a proper choice of $l$, $\widehat{\rho}_is_i$ as an estimator of $\rho_is_i$ has an approximation guarantee for all $i\in V$. Specifically, we establish an $(\epsilon,\delta)$-approximation of $\widehat{\rho}_is_i$: for any small parameters $\epsilon>0$ and $\delta>0$, the approximation error $\widehat{\rho}_is_i$  is bounded by $\epsilon$ with probability  at least $1-\delta$. Theorem~\ref{th-Fast} shows how to properly choose $l$ so that $\widehat{\rho}_is_i$  is an $(\epsilon,\delta)$-approximation of $\rho_is_i$.

\begin{figure*}[t]
	\centering
	\includegraphics[width=1.9\columnwidth]{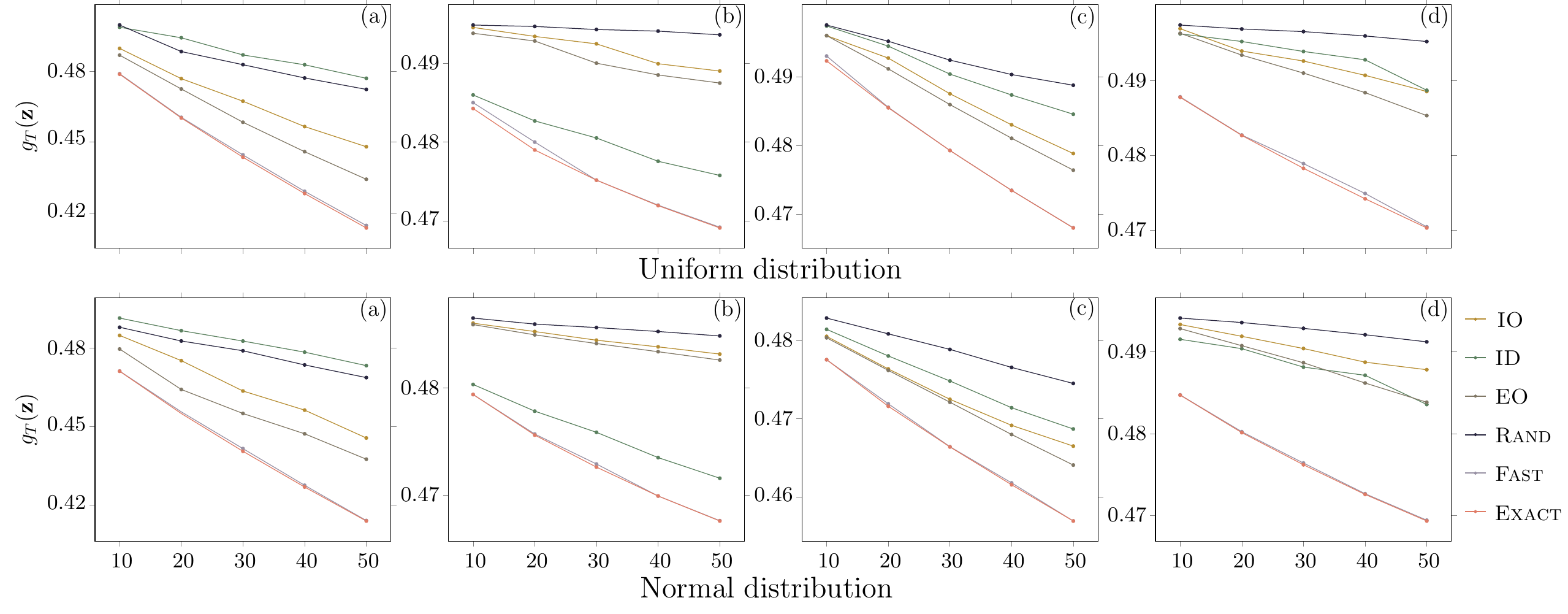}
	\caption{Average of equilibrium expressed opinions for our two algorithms \textsc{Exact}  and \textsc{Fast}, and four baseline heuristics \textsc{Random}  (Rand), \textsc{In-degree} (ID), \textsc{Internal opinion} (IO), and \textsc{Expressed opinion} (EO), on four directed real networks: (a) Filmtrust, (b) Dblp, (c) Humanproteins, and (d) P2p-Gnutella08.}\label{effectiveness}	
\end{figure*}

\begin{theorem}\label{th-Fast}
	For any  $ \epsilon>0 $  and $ \delta\in(0,1) $, if $l$ is chosen obeying   $ l =\left \lceil \frac{1}{2\epsilon^2}\ln{\frac{2}{\delta} }  \right \rceil   $, then for any  $i \in V$,  $ \mathbb{P}\left\{ |\widehat{\rho}[i]s_i- \rho_is_i| >\epsilon \right\} \leq \delta $.
%	\begin{equation}\label{key}
%		\mathbb{P}\left\{ |\widehat{\rho}[i]s_i- \rho_is_i| >\epsilon \right\} \leq \delta.
%	\end{equation}
\end{theorem}

Recall that our problem aims to determine the optimal set $T$, which consists of $k$ nodes with the largest $\rho_is_i$. To avoid calculating $\rho_i$ directly, we propose a fast algorithm (Algorithm~\ref{al-Fast}), which returns a set $\widehat{T}$ containing top $k$ nodes of the highest $\hat{\rho}[i]s_i$.  Based on the result of Theorem~\ref{th-Fast}, we can get a union bound between $g_{\widehat{T}}(\boldsymbol{\mathit{z}})$ and $g_{T}(\boldsymbol{\mathit{z}})$, as stated in the following theorem.

\begin{theorem}
    For given  parameters $k, \epsilon, \delta$, if $l$ is chosen according to Theorem~\ref{th-Fast}, the inequality $|g_{\widehat{T}} (\boldsymbol{\mathit{z}})-g_{T}(\boldsymbol{\mathit{z}})|<2k\epsilon $ holds with high probability. 
\end{theorem}

\begin{proof}
	According to Theorem 5.4, we suppose now inequalities $|\widehat{\rho}[i]s_i- \rho_is_i| >\epsilon$ hold for any $i\in V$.
	Since the nodes in set $ T $ have the top value of $ \rho_is_i $, we have
	$$g_{\widehat{T}}(\boldsymbol{\mathit{z}})-g_{T}(\boldsymbol{\mathit{z}}) = \sum_{i\notin \widehat{T}}\rho_is_i-\sum_{i\notin T}\rho_is_i=\sum_{i\in {T}}\rho_is_i-\sum_{i\in \widehat{T}}\rho_is_i \geq 0.$$
	By Theorem 5.4, one obtains
	\begin{align*}
		&\quad g_{\widehat{T}}(\boldsymbol{\mathit{z}})-g_{T}(\boldsymbol{\mathit{z}})\leq \sum_{i\in {T}}\widehat{\rho}[i]s_i-\sum_{i\in \widehat{T}}\rho_is_i  +k\epsilon  \\& \leq \sum_{i\in \widehat{T}}\widehat{\rho}[i]s_i-\sum_{i\in \widehat{T}}\rho_is_i  +k\epsilon \leq 2k\epsilon,
	\end{align*} which completes the proof.
\end{proof}	

Therefore, for any fixed $k$, the number of samples does not depend on $n$.

\section{Experiments}

In this section, we conduct extensive experiments on various real-life directed networks, in order to evaluate the performance of our two algorithms \textsc{Exact} and \textsc{Fast} in terms of  effectiveness and efficiency. The data sets of selected real networks are publicly available in the KONECT~\cite{Ku13} and SNAP~\cite{LeSo16}, the detailed information of which is presented in the first three columns of Table~\ref{NetInf}. In the dataset networks, the number $n$ of nodes ranges from about 1 thousand to 24 million, and the number $m$ of directed edges ranges from about 2 thousand to 58 million. All our experiments are programmed in Julia using a single thread, and are run on a machine equipped with 4.2 GHz Intel i7-7700 CPU and 32GB of main memory. 
%The source code is publicly available on \url{https://github.com/OpinionOptimization/Opinion}.

\subsection{Effectiveness}

We first compare the effectiveness of our algorithms \textsc{Exact} and \textsc{Fast} with four baseline schemes for node selection: \textsc{Random}, \textsc{In-degree}, \textsc{Internal opinion}, and \textsc{Expressed opinion}. \textsc{Random} selects $k$ nodes at random. \textsc{In-degree} chooses $k$ nodes with the largest in-degree, since a node with a high in-degree may has a strong influence on other nodes~\cite{XuHuWu20}. For \textsc{Internal opinion} and \textsc{Expressed opinion}, they have been used in~\cite{ GiTeTs13}. \textsc{Internal opinion} returns $ k $ nodes with the largest original internal opinions, while \textsc{Expressed opinion} selects $ k $ nodes with the largest equilibrium expressed opinions in the FJ model corresponding to the original internal opinion vector. 

\begin{table*}[t]
\fontsize{9.8pt}{10.8pt} \selectfont
        \centering
	\resizebox{0.95\textwidth}{!}{
	\begin{tabular}{llllllllll}
		\hline
		\multirow{2}{*}{Network} & \multirow{2}{*}{Nodes} & \multirow{2}{*}{Arcs} & \multicolumn{4}{l}{Running time ($s$) for \textsc{Exact} and \textsc{Fast}} & \multicolumn{3}{l}{Relative error ($\times 10^{-3}$)} \\ \cline{4-10} 
		&                        &                        & Optimal         & $ l = 500  $      & $ l = 1000  $      & $ l=2000   $       & $ l = 500  $     & $ l = 1000  $     & $ l=2000  $    \\ \hline
		Filmtrust                & 874                    & 1,853                  & 0.023         & 0.019         & 0.037          & 0.042          & 1.08         & 0.55          & 0.11       \\
	Humanproteins            & 2,239                  & 6,452                  & 0.251         & 0.032         & 0.039          & 0.056          & 0.19        & 0.16          & 0.03       \\
	Adolescenthealth         & 2,539                  & 12,969                 & 0.354         & 0.034         & 0.068          & 0.134          & 0.72         & 0.59         & 0.09       \\
	P2p-Gnutella08           & 6,301                  & 20,777                 & 4.825         & 0.067         & 0.123          & 0.244          & 0.83        & 0.64        & 0.04      \\
	Wiki-Vote                & 7,115                  & 103,689                & 7.405         & 0.078         & 0.156          & 0.312          & 1.13         & 0.87         & 0.13       \\
	Dblp                     & 12,590                 & 49,744                 & 40.870        & 0.106         & 0.210          & 0.419          & 0.38         & 0.13          & 0.08       \\
	Wikipedialinks           & 17,649                 & 296,918                & 110.744       & 0.248         & 0.477          & 0.932          & 1.32         & 0.97          & 0.06       \\
	Twitterlist              & 23,370                 & 33,101                 & 259.484       & 0.127         & 0.250          & 0.498          & 0.25         & 0.12         & 0.01     \\
	P2p-Gnutella31           & 62,586                 & 147,892                & -             & 0.628         & 1.236          & 2.550          & -            & -             & -          \\
	Soc-Epinions             & 75,879                 & 508,837                & -             & 1.260         & 2.501          & 4.973          & -            & -             & -          \\
	Email-EuAll              & 265,009                & 418,956                & -             & 3.016         & 5.929          & 11.844         & -            & -             & -          \\
	Stanford                 & 281,903                & 2,312,500              & -             & 7.474         & 14.908         & 29.815         & -            & -             & -          \\
	NotreDame                & 325,729                & 1,469,680              & -             & 4.823         & 9.574          & 19.232         & -            & -             & -          \\
	BerkStan                 & 685,230                & 7,600,600              & -             & 14.021        & 28.009         & 56.130         & -            & -             & -          \\
	Google                   & 875,713                & 5,105,040              & -             & 26.583        & 53.655         & 106.005        & -            & -             & -          \\
	NorthwestUSA             & 1,207,940              & 2,820,770              & -             & 27.758        & 55.509         & 110.410        & -            & -             & -          \\
	WikiTalk                 & 2,394,380              & 5,021,410              & -             & 20.277        & 37.622         & 75.105         & -            & -             & -          \\
	Greatlakes               & 2,758,120              & 6,794,810              & -             & 64.391        & 128.167        & 255.147        & -            & -             & -          \\
	FullUSA                  & 23,947,300             & 57,708,600             & -             & 559.147       & 1116.550       & 2230.770       & -            & -             & -          \\ \hline
	\end{tabular}}
 \caption{The running time  and the relative error of Algorithms~\ref{al-optimal} and~\ref{al-Fast} on real networks for various sampling number $ l $. }\label{NetInf}
\end{table*}

In our experiment, the number $l$ of samplings in algorithm \textsc{Fast} is set be $500 $. For each node $i$, its internal opinion $s_i$ is generated uniformly in the interval $[0,1]$. For each real network, we first calculate the equilibrium expressed opinions of all nodes and their average opinion for the original internal opinions. Then, using our algorithms \textsc{Exact} and \textsc{Fast} and the four baseline strategies, we select $ k = 10,20,30,40,50 $ nodes and change their internal opinions to 0, and recompute the average expressed opinion associated with the modified internal opinions. We also execute experiments for other distributions of internal opinions. For example, we consider the case that the internal opinions follow a normal distribution with mean $ 0 $ and variance $ 1 $. For this case, we perform a linear transformation, mapping the internal opinions into interval $[0,1]$, so that the smallest internal opinion is mapped to 0, while the largest internal opinion corresponds to 1. 
 As can be seen from Figure~\ref{effectiveness},  for each network algorithm \textsc{Fast} always returns a result close to the optimal solution corresponding to algorithm \textsc{Exact} for both  uniform distribution and standardized normal distribution, outperforming the four other baseline strategies.

For the cases that internal opinions obey power-law distribution or exponential distribution, here we do not report the results since they are similar to that observed in Figure~\ref{effectiveness}.

\subsection{Efficiency and Scalability}

As shown above, algorithm \textsc{Fast} has similar effectiveness to that of algorithm \textsc{Exact}. Below we will show that algorithm \textsc{Fast} is more efficient than algorithm \textsc{Exact}. To this end, in Table~\ref{NetInf} we compare the performance of algorithms \textsc{Exact} and \textsc{Fast}. First, we compare the running time of the two algorithms on the real-life directed networks listed in Table~\ref{NetInf}. For our experiment, the internal opinion of all nodes in each network obeys uniform distribution from $[0,1]$, $k$ is equal to 50, and $l$ is chosen to be 500, 1000, and 2000. As shown in Table~\ref{NetInf}, \textsc{Fast} is significantly faster than \textsc{Exact} for all $l$, which becomes more obvious when the number of nodes increases. For example, \textsc{Exact} fails to run on the last 11 networks in Table~\ref{NetInf}, due to time and memory limitations. In contrast, \textsc{Fast} still works well in these networks. Particularly, algorithm \textsc{Fast} is scalable to massive networks with more than twenty million nodes, e.g., FullUSA with over $ 2.9\times 10^7$ nodes.

%Next, we will compare the performance of algorithms \textsc{Optimal} and \textsc{Fast}. For this purpose, we compare the running time of algorithms \textsc{Optimal} and \textsc{Fast} on real-life directed networks in Table 1. For each network,  the internal opinions obey uniform distribution from $ (0,1) $, and we set $ k=50 $ in this experiment. Table 1 shows  \textsc{Fast} is much faster than \textsc{Optimal}, which is more obvious for larger networks. In particular, \textsc{Fast} is scalable to massive networks with over fifty million nodes. For example, for the last 12 networks in Table 1, such as FullUSA with over $ 2.9\times 10^7$ nodes, \textsc{Optimal} can’t run due to the memory limitation, while \textsc{Fast} still works well.

Table~\ref{NetInf} also reports quantitative comparison of the effectiveness between algorithms \textsc{Exact} and \textsc{Fast}. Let $g_T $ and $g_{\widehat{T}}$ denote the average opinion obtained, respectively, by algorithms \textsc{Exact} and \textsc{Fast}, and let $ \gamma = |g_T-g_{\widehat{T}} |/g_T  $ be the relative error of $g_{\widehat{T}}$ with respect to $ g_T$. The last three columns of Table~\ref{NetInf} present the relative errors for different real networks and various numbers $l$ of samplings. From the results, we can see that for all networks and different $l$, the relative error $\gamma$ is negligible, with the largest value being less than $0.0014$. Moreover, for each network, $\gamma$ is decreased when $l$ increases. This again indicates that the results returned by \textsc{Fast} are very close to those corresponding to \textsc{Exact}. Therefore, algorithm \textsc{Fast} is both effective and efficient, and scales to massive graphs.

%We proceed to compare the effectiveness of algorithms \textsc{Optimal}and \textsc{Fast}. We define $ g_T $ and $ g_{\widehat{T}} $ as the average opinion got by  algorithms \textsc{Optimal} and \textsc{Fast}, and define $ \gamma = |g_T-g_{\widehat{T}} |/g_T  $ as the relative error. The results of relative errors for different real networks and various number of sampling  are presented in Table 1, which shows that relative errors $ \gamma $ are very small, with the largest value equal to 0.83\textperthousand ~  when $ l=2000 $. Thus, the results returned by \textsc{Fast} are very close to those returned by \textsc{Optimal}, implying that \textsc{Fast} is both effective and efficient.

\section{Conclusions}

In this paper, we studied how to optimize social opinions based on the Friedkin-Johnsen (FJ) model in an unweighted directed social network with $n$ nodes and $ m $ edges, where the internal opinion $s_i$, $i =1,2,\cdots,n$, of every node $i$ is in interval $[0, 1]$. We concentrated on the problem of minimizing the average of equilibrium opinions by selecting a set $U$ of $ k\ll n $ nodes and modifying their internal opinions to 0. Although the problem seems combinatorial, we proved that there is an algorithm \textsc{Exact} solving it in $O(n^3)$ time, which returns the $k$ optimal nodes with the top $k$ values of $\rho_i s_i $, $ i =1,\cdots,n $, where $\rho_i$ is the structure centrality of node $i$.

Although algorithm \textsc{Exact} avoids the naïve enumeration of all $\tbinom{n}{k}$ cases for set $U$, it is not applicable to large graphs. To make up for this deficiency, we proposed a fast algorithm for the problem. To this end, we provided an interpretation of $\rho_i$ in terms of rooted spanning converging forests, and designed a fast sampling algorithm \textsc{Fast} to estimate $\rho_i$ for all nodes by using a variant of Wilson's Algorithm. The algorithm simultaneously returns $k$ nodes with largest values of $\rho_i s_i $ in $ O(ln) $ time, where $ l $ denotes the number of samplings. Finally, we performed experiments on many real directed networks of different sizes to demonstrate the performance of our algorithms. The results show that the effectiveness of algorithm \textsc{Fast} is comparable to that of algorithm \textsc{Exact}, both of which are better than the baseline algorithms. Furthermore, relative to \textsc{Exact}, \textsc{Fast} is more efficient, since is  \textsc{Fast} is scalable to massive graphs with over twenty million nodes, while \textsc{Exact} only applies to graphs with less than tens of thousands of nodes. It is worth mentioning that it is easy to extend or modify  our algorithm to weighed digraphs and apply it to solve other optimization problems for opinion dynamics.

\section*{Acknowledgements}

Zhongzhi Zhang is the corresponding author.
This work was supported by the Shanghai
Municipal Science and Technology Major Project (No. 2018SHZDZX01), the National Natural Science
Foundation of China (Nos. 61872093 and U20B2051), ZJ Lab, and Shanghai Center
for Brain Science and Brain-Inspired Technology.

\fontsize{9.8pt}{10.8pt} \selectfont
\bibliography{opinion}

\end{document}